\documentclass[11pt]{book}

\usepackage[dvips]{epsfig}

\usepackage{amssymb}
\usepackage{multirow}
\usepackage{amsbsy}
\usepackage{pstricks}
\usepackage{soul}

\newenvironment{proof}{\paragraph*{Proof:}}{\hfill$\square$}

\catcode`\@=11

\def\@normalsize{\@setsize\normalsize{13pt}\xipt\@xipt
	\abovedisplayskip 11pt plus3pt minus6pt
	\belowdisplayskip \abovedisplayskip
	\abovedisplayshortskip \z@ plus3pt
	\belowdisplayshortskip 6.6pt plus3.5pt minus3pt} 

\def\small{\@setsize\small{12pt}\xipt\@xipt
	\abovedisplayskip 10pt plus2pt minus5pt
	\belowdisplayskip \abovedisplayskip
	\abovedisplayshortskip \z@ plus3pt
	\belowdisplayshortskip 6pt plus3pt minus3pt
	\def\@listi{\topsep 6pt plus 2pt minus 2pt
		\parsep 3pt plus 2pt minus 1pt
		\itemsep \parsep}}

\def\footnotesize{\@setsize\footnotesize{10pt}\ixpt\@ixpt
	\abovedisplayskip 8pt plus 2pt minus 4pt
	\belowdisplayskip \abovedisplayskip
	\abovedisplayshortskip \z@ plus 1pt
	\belowdisplayshortskip 4pt plus 2pt minus 2pt
	\def\@listi{\topsep 4pt plus 2pt minus 2pt
		\parsep 2pt plus 1pt minus 1pt
		\itemsep \parsep}}

\def\scriptsize{\@setsize\scriptsize{9.5pt}\viiipt\@viiipt}
\def\tiny{\@setsize\tiny{7pt}\vipt\@vipt}
\def\large{\@setsize\large{14pt}\xiipt\@xiipt}
\def\Large{\@setsize\Large{18pt}\xivpt\@xivpt}
\def\LARGE{\@setsize\LARGE{22pt}\xviipt\@xviipt}
\def\huge{\@setsize\huge{25pt}\xxpt\@xxpt}
\def\Huge{\@setsize\Huge{30pt}\xxvpt\@xxvpt}

\def\section{\@startsection {section}{1}{\z@}%
	{-1.5\baselineskip plus-1pt minus-3pt}{1\baselineskip plus1pt minus2pt}%
	{\centering\normalsize\bf}}
\def\subsection{\@startsection{subsection}{2}{\z@}%
	{-1\baselineskip plus-1pt minus-2pt}{1\baselineskip plus1pt minus2pt}%
	{\normalsize\sc\noindent}}
\def\subsubsection{\@startsection{subsubsection}{3}{\z@}%
	{-1\baselineskip plus-1pt minus-2pt}{1sp}{\normalsize\it\noindent}}
\def\paragraph{\@startsection{paragraph}{4}{\z@}%
	{1\baselineskip plus1pt minus2pt}{-1em}{\normalsize\it\noindent}}
\let\subparagraph=\paragraph

\def\tableofcontents{\@restonecolfalse\if@twocolumn\@restonecoltrue
	\onecolumn\fi\OSIDcont\@starttoc{con}\if@restonecol\twocolumn\fi}

\def\l@section{\@dottedtocline{1}{0em}{.66em}}

\def\thebibliography#1{\section*{{Bibliography}\@mkboth
		{BIBLIOGRAPHY}{BIBLIOGRAPHY}}\footnotesize\rm\list
	{[\arabic{enumi}]}{\settowidth\labelwidth{[#1]}\leftmargin\labelwidth
		\advance\leftmargin\labelsep\usecounter{enumi}}
	\def\newblock{\hskip .11em plus .33em minus -.07em}
	\sloppy\clubpenalty4000\widowpenalty4000
	\sfcode`\.=1000\relax}

\def\ps@myheadings{\let\@mkboth\@gobbletwo
	\def\@oddhead{\hfil{\footnotesize\rm\rightmark}\hfil}
	\def\@evenhead{\hfil{\footnotesize\rm\leftmark}\hfil}
	\def\@oddfoot{\hfil{\footnotesize\sf\artid-\thepage}\hfil}
	\def\@evenfoot{\hfil{\footnotesize\sf\artid-\thepage}\hfil}
	\def\sectionmark##1{}\def\subsectionmark##1{}}

\newcounter{paPer}     %
\def\ps@osiD{\let\@mkboth\@gobbletwo
	\def\@oddfoot{\hfil{\footnotesize\sf\artid-\thepage}\hfil}
	\def\@evenhead{}\let\@evenfoot\@oddfoot}

\def\cite{\@ifnextchar [{\@tempswatrue\@Rcitex}{\@tempswafalse\@Rcitex[]}}

\def\@Rcitex[#1]#2{\if@filesw\immediate\write\@auxout{\string\citation{#2}}\fi
	\def\@citea{}\@cite{\@for\@citeb:=#2\do
		{\@citea\def\@citea{,\penalty\@m\,}\@ifundefined
			{b@\the\value{paPer}R\@citeb}{{\bf ?}\@warning
				{Citation `\@citeb' on page \thepage \space undefined}}%
			\hbox{\csname b@\the\value{paPer}R\@citeb\endcsname}}}{#1}}

\long\def\@caption#1[#2]#3{\par\addcontentsline{\csname
		ext@#1\endcsname}{#1}{\protect\numberline{\csname
			the#1\endcsname}{\ignorespaces #2}}\begingroup
	\@parboxrestore
	\small                                        
	\@makecaption{\csname fnum@#1\endcsname}{\ignorespaces #3}\par
	\endgroup}

\catcode`\@=12

\let\Rlabel=\label
\let\Rbibitem=\bibitem
\let\Rref=\ref
\let\Rpageref=\pageref
\def\label#1{\expandafter\Rlabel{\the\value{paPer}R#1}}
\def\bibitem#1{\expandafter\Rbibitem{\the\value{paPer}R#1}}
\def\ref#1{\expandafter\Rref{\the\value{paPer}R#1}}
\def\pageref#1{\expandafter\Rpageref{\the\value{paPer}R#1}}

\def\thesection{\arabic{section}.}

\def\YYMm{\rule{0ex}{4em}}
\newtoks\TITsi
\newtoks\TITsii

\def\title#1{\def\TITs{\LARGE{\raggedright\noindent\YYMm #1%
			\vskip8pt\par}}}

\def\author#1{\autMM{#1}\def\LHD{#1}}
\def\and{{\rm\lowercase{and}}}

\def\autMM#1{\TITsii={\vskip10pt\par\normalsize\rm\noindent #1\par}%
\TITsi=\expandafter{\TITs}\edef\TITs{\the\TITsi\the\TITsii}}

\def\address#1{\TITsii={\vskip6pt\par\footnotesize\sl\noindent #1\par}%
\TITsi=\expandafter{\TITs}%
\edef\TITs{\the\TITsi\the\TITsii}}

\def\received#1{\TITsii={\vskip10pt\par\small\rm\noindent(Received: #1)\par}%
\TITsi=\expandafter{\TITs}\edef\TITs{\the\TITsi\the\TITsii}}

\def\headtitle#1{\def\RHD{#1}}
\def\headauthor#1{\def\LHD{#1}}

\def\abst{{\bf Abstract.}}
\def\abstract#1{\TITs
\vskip15pt\par\noindent
{\footnotesize{\abst~} #1\vskip3pt\par}
\markright{\RHD}
\markboth{\LHD}{\RHD}}

\def\startpaper{%
\cleardoublepage
\setcounter{section}{0}
\stepcounter{paPer}
\setcounter{equation}{0}
\setcounter{footnote}{0}
\setcounter{figure}{0}
\setcounter{table}{0}
\def\theequation{\arabic{equation}}
\def\thefootnote{\arabic{footnote}}
\setcounter{defn}{0}
\setcounter{thm}{0}
\setcounter{lem}{0}
\setcounter{prop}{0}
\setcounter{rem}{0}
\thispagestyle{osiD}}

\def\OSIDcont{\cleardoublepage\thispagestyle{empty}
\markright{}\markboth{}{}
\normalsize\rm
\hspace*{\fill}{\large\rm
Contents of the Volume \Volume, Number \Number}\hspace*{\fill}
\par\vspace{1.5em}
\par\noindent}

\def\endpaper{\expandafter\label{\the\value{paPer}OpSy}}

\def\1{{\mathchoice{\rm 1\mskip-4mu l}{\rm 1\mskip-4mu l}%
{\rm 1\mskip-4.5mu l}{\rm 1\mskip-5mu l}}}

\def\varkappa{\mbox{\bBB\char 123}}

\def\longhookrightarrow{\lhook\joinrel\relbar\joinrel\rightarrow}

\def\longhookUp{\lower6pt\hbox{\rotatebox{90}{$\longhookrightarrow$}}}

\newtheorem{thm}{\rm THEOREM}

\newtheorem{defn}{\rm DEFINITION}
\newtheorem{exmp}{\rm EXAMPLE}
\newtheorem{rem}{\it Remark}

\def\theequation{\thesection\arabic{equation}}

\def\Myskip{\setlength{\baselineskip}{13pt}}

\def\text#1{\quad\mbox{\rm  #1 }\quad}

\usepackage{physics}
\def\DOInumber{}

\input xy
\xyoption{all}

\InputIfFileExists{psfig.sty}{\typeout{^^Jpsfig.sty inputed...ok}}{\typeout{^^JWarning: psfig.sty could not be found.^^J}}
\usepackage{orcidlink}
\usepackage{amsmath}
\usepackage{hyperref}

\begin{document}

\def\artid{0000001}
\def\Volume{15}
\def\Number{1}
\def\Year{2008}
\setcounter{page}{1}

\def\DOInumber{}

\startpaper

\newcommand{\Mn}{M_n(\mathbb{C})}
\newcommand{\Mk}{M_k(\mathbb{C})}
\newcommand{\id}{\mbox{id}}
\newcommand{\ot}{{\,\otimes\,}}
\newcommand{{\Cd}}{{\mathbb{C}^d}}
\newcommand{\sbsigma}{{\mbox{\scriptsize \boldmath $\sigma$}}}
\newcommand{\sbalpha}{{\mbox{\scriptsize \boldmath $\alpha$}}}
\newcommand{\sbbeta}{{\mbox{\scriptsize \boldmath $\beta$}}}
\newcommand{\bsigma}{{\mbox{\boldmath $\sigma$}}}
\newcommand{\balpha}{{\mbox{\boldmath $\alpha$}}}
\newcommand{\bbeta}{{\mbox{\boldmath $\beta$}}}
\newcommand{\bmu}{{\mbox{\boldmath $\mu$}}}
\newcommand{\bnu}{{\mbox{\boldmath $\nu$}}}
\newcommand{\ba}{{\mbox{\boldmath $a$}}}
\newcommand{\bb}{{\mbox{\boldmath $b$}}}
\newcommand{\sba}{{\mbox{\scriptsize \boldmath $a$}}}
\newcommand{\MD}{\mathfrak{D}}
\newcommand{\sbb}{{\mbox{\scriptsize \boldmath $b$}}}
\newcommand{\sbmu}{{\mbox{\scriptsize \boldmath $\mu$}}}
\newcommand{\sbnu}{{\mbox{\scriptsize \boldmath $\nu$}}}
\def\oper{{\mathchoice{\rm 1\mskip-4mu l}{\rm 1\mskip-4mu l}%
		{\rm 1\mskip-4.5mu l}{\rm 1\mskip-5mu l}}}
\def\<{\langle}
\def\>{\rangle}
\def\theequation{\thesection\arabic{equation}}

\thispagestyle{plain}

\title{Systematic construction of ROCN Bell-inequalities}
\author{Arturo Konderak$^{1,*}$\orcidlink{0000-0002-4546-2626}, Patryk Michalski$^{1,2,\dagger}$\orcidlink{0009-0009-0305-7356}}
\address{$^1$Center for Theoretical Physics, PAS, al. Lotnik\'{o}w 32/46, 02-668 Warsaw, Poland} 
\address{$^2$Institute of Theoretical Physics, University of Warsaw, Pasteura 5, 02-093 Warsaw, Poland}
\address{$^*$email: \href{mailto:akonderak@cft.edu.pl}{\texttt{akonderak@cft.edu.pl}}}
\address{$^\dagger$email: \href{mailto:pmichalski@cft.edu.pl}{\texttt{pmichalski@cft.edu.pl}}}
\headauthor{A. Konderak, P. Michalski}
\headtitle{Systematic construction of ROCN Bell-inequalities}
\received{\today}

\abstract{Self-testing constitutes one of the most powerful forms of device certification, enabling a complete and device-independent characterization of a quantum apparatus solely from the observed correlations. In recent work by the authors~\cite{Michalski2025}, a general framework was introduced for constructing Bell inequalities that self-test entire families of Clifford generators. In this manuscript, we develop an alternative and complementary self-testing criterion based on symmetric spanning sets. This formulation provides an explicit and constructive route to designing self-testing Bell inequalities in arbitrary dimensions.}

\Myskip



\section{Introduction}

Bell non-locality is one of the most distinguishing features of quantum theory. It provides a rigorous and experimentally testable certification of the quantum nature of correlations between spatially separated systems. {Consider, for example,} two distant parties, {Alice} and {Bob}, who share a quantum resource---typically an entangled state. When each party performs local measurements on their respective subsystems, the resulting joint statistics can, in certain scenarios, exhibit correlations stronger than those permitted by any classical model based on local hidden variables~\cite{Augusiak2014}. In such cases, the correlations are said to be \emph{non-local}~\cite{BELL2001}.

The importance of non-locality is difficult to overstate, as it underpins a broad range of quantum certification protocols, including device-independent (DI) randomness certification~\cite{PhysRevA.93.040102}, quantum cryptography~\cite{Scarani2012}, quantum communication complexity~\cite{Buhrman2010}, and the self-testing of quantum devices~\cite{Supic2020}. Experimental investigations of Bell non-locality date back to the 1970s~\cite{Freedman72}, but it was not until 2015 that the first loophole-free demonstrations were achieved~\cite{Hensen2015}.

One of the most fundamental tools for certifying non-locality is the violation of a Bell inequality, which any local model must satisfy. Beyond that, Bell inequalities can also be used to characterize, for instance, full non-locality~\cite{Barrett2006,Salavrakos2017}, enable device-independent quantum cryptography~\cite{Acin2007,Ekert1992}, certify the dimensionality of a quantum channel~\cite{Engineer2025}, or demonstrate the infiniteness of the number of measurements~\cite{Tendick2025}. Viewing the Hamiltonian of a multipartite system as a Bell operator, one can prove that the ground state exhibits non-local properties~\cite{Tura2017}Remarkably, there exist cases in which the maximal violation of a Bell inequality also certifies the underlying experimental apparatus---that is, both the shared quantum state and the measurements performed by the parties~\cite{Supic2020}. In such {situations}, the Bell inequality is said to \emph{self-test} a reference quantum strategy.

The seminal CHSH inequality~\cite{Clauser1969} provides the simplest, and still one of the most powerful, examples of a Bell inequality. In this scenario, each party performs only two dichotomic measurements, and the maximal violation of the CHSH inequality self-tests the maximally entangled qubit state. Several generalizations of CHSH have been proposed for higher-dimensional systems. Notably, the Elegant Bell inequality~\cite{Gisin2009} achieves its maximal violation by employing measurements derived from SIC-POVMs, while other extensions include the Hadamard Bell inequalities~\cite{Goyeneche2023}. Additionally, a general class of Bell inequalities based on Platonic solids was introduced in Ref.~\cite{Pal2022}. More recently, the ROCN class of Bell inequalities {was proposed in Ref.}~\cite{Michalski2025}, for which the quantum bound can be derived analytically. Under suitable structural conditions, the maximal violation of ROCN inequalities allows for the self-testing of maximally entangled states in higher dimensions.

In this article, we further investigate the structure of ROCN Bell inequalities. In particular, Ref.~\cite{Michalski2025} provides a necessary and sufficient condition for a Bell inequality to self-test the reference strategy. However, this condition can appear abstract from a physical perspective. We therefore begin by providing a more intuitive, physically motivated interpretation of this criterion. Moreover, by analyzing the structure of symmetric Hilbert spaces, we derive a sufficient condition for an ROCN Bell inequality to be self-testing: namely, the columns of the correlation matrix should form a \emph{symmetric spanning set}. This result enables a systematic construction of families of ROCN Bell inequalities that self-test any even number of Majorana observables. Such constructions are particularly relevant given the central role of Majorana observables and Clifford groups across physics, from quantum thermodynamics~\cite{Araki2003} and quantum field theory~\cite{Derezinski2023} to quantum computing~\cite{Nielsen2012,Romanova2025}.

The paper is organized as follows. In Section~\ref{sec:preliminaries}, we provide an introduction to Bell non-locality, with a focus on ROCN Bell inequalities (Subsection~\ref{sec:rocn_Bell_inequalities}) and self-testing (Subsection~\ref{sec:rocb_self_testing}). Section~\ref{sec:symmetric_subspace} recalls the definition of the symmetric two-fold tensor product and introduces the notion of a symmetric spanning set. In Theorem~\ref{th:set_spanning}, we present an explicit family of vectors whose two-fold tensor products span the full symmetric subspace. Section~\ref{sec:sufficient_condition} provides a {physically transparent criterion for} self-testing within the ROCN framework. Combining this result with Theorem~\ref{th:set_spanning}, Section~\ref{sec:self-testing_construction} describes a construction of ROCN Bell inequalities that self-test any even number of Clifford generators. Finally, in Section~\ref{sec:examples}, we illustrate the procedure with explicit examples involving two and four Clifford generators.

\section{Preliminaries}\label{sec:preliminaries}

In this section, we present the main concepts related to Bell non-locality. In Subsection~\ref{sec:rocn_Bell_inequalities}, we introduce the notion of Bell inequalities, and in Subsection~\ref{sec:rocb_self_testing}, we discuss the concept of self-testing. Throughout, the discussion is specialized to ROCN Bell inequalities.

\subsection{ROCN Bell inequalities}\label{sec:rocn_Bell_inequalities}

We consider a bipartite scenario in which two parties, Alice and Bob, have access to local Hilbert spaces $\mathcal H_{\mathrm{A}}$ and $\mathcal H_{\mathrm{B}}$, respectively, and share a quantum state
\begin{equation}
	\ket{\Psi_{\mathrm{AB}}}\in\mathcal H_{\mathrm{A}}\otimes \mathcal H_{\mathrm{B}}.
\end{equation}
Each party performs local measurements. We focus on the case where all measurements have two outcomes, labeled by $\pm 1$. In this setting, Alice’s and Bob’s measurements can be represented by self-adjoint observables $A_{i}$ for $i=1,\dots, m$ and $B_j$ for $j=1,\dots, n$, satisfying $A_i^2=B_j^2=\mathbb I$.

The quantity $p(a,b|i,j)$, denoting the probability of obtaining outcomes $a,b=\pm 1$ upon measuring the observables $A_i$ and $B_j$, is referred to as a quantum correlation. It can be expressed through an affine transformation in terms of expectation values~\cite{Goh2018}:
\begin{equation}
	p(a,b|i,j)=\frac{1}{4}\left[1+(-1)^{\frac{1-a}{2}}\langle A_i \rangle + (-1)^{\frac{1-b}{2}} \langle B_j \rangle - (-1)^{\frac{a+b}{2}} \langle A_i B_j\rangle \right].
\end{equation}
Thus, quantum correlations may equivalently be studied in the space of probability distributions or in terms of expectation values. 

A Bell inequality provides a criterion to determine whether a given correlation exhibits nonlocality. Such inequalities are formulated in terms of a Bell operator, which, for correlation-type Bell inequalities, takes the form
\begin{equation}
	\mathcal B_h:=\sum_{i=1}^m\sum_{j=1}^n h_{ij}A_i \otimes B_j.
\end{equation}
Note that for correlation-type Bell operators, no linear terms in the local observables $\langle A_i\rangle$ or $\langle B_j\rangle$ appear. The coefficients $h_{ij}$ form a real $m\times n$ matrix, and the associated Bell inequality is defined as
\begin{equation}
	I_h:= \langle \mathcal B_h\rangle= \sum_{i=1}^m\sum_{j=1}^n h_{ij}\langle A_iB_j\rangle\leqslant \beta_C^h,
\end{equation}
where $\beta_C^h$ denotes the \textit{classical} or \textit{local bound}. This bound corresponds to the maximal value of $I_h$ over all correlations that factorize as $\langle A_i B_j\rangle = \langle A_i\rangle\langle B_j\rangle$ with $\langle A_i\rangle, \langle B_j\rangle \in \{-1,+1\}$. We refer to $I_h$ as the \textit{Bell functional}. For a given quantum realization, if $I_h > \beta_C$, the corresponding correlations are said to be nonlocal. The \textit{quantum bound}, denoted by $\beta_Q^h$, is the maximal attainable value of the Bell functional $I_h$ over all possible quantum realizations. 

The family of ROCN Bell inequalities is defined by requiring that the coefficient matrix $h$ satisfy the ROCN properties introduced in Ref.~\cite{Michalski2025}. For completeness, we recall the definition below.

\begin{defn}[ROCN matrix]
	Let $h$ be a real $m\times n$ matrix with $m\leqslant n$. We say that $h$ is a \emph{row-orthogonal and column-normalized (ROCN) matrix} if it satisfies the following conditions:
	\begin{align}
		&\forall\; i, k \in \{1,\ldots,m\}: \quad \sum_{j = 1}^n h_{ij} h_{kj} = \delta_{ik} \sum_{j = 1}^n \label{eq:hmatrix_1}h_{ij}^2 \quad \text{(orthogonality of rows)}, \\ \label{eq:hmatrix_2}
		&\forall\; j \in \{1,\ldots,n\}: \quad \quad \sum_{i = 1}^m h_{ij}^2 = 1 \quad \text{(normalization of columns)}.
	\end{align}
	We additionally assume that each row contains at least one non-zero entry.
\end{defn}

Examples of ROCN matrices include orthogonal matrices (for $m=n$) or semiorthogonal matrices. Notably, both the CHSH and Elegant Bell inequalities (EBI) can be expressed in ROCN form. The quantum bound $\beta_Q^h$ associated with a ROCN Bell inequality can be computed explicitly:
\begin{equation}
	\beta_Q^h = \sup_{\ket{\psi},\{A_i\},\{B_j\}}\mel{\psi}{\mathcal B_h}{\psi}=n.
\end{equation}
This bound follows from a sum-of-squares (SOS) decomposition; see Ref.~\cite{Michalski2025} for details. 

The maximal quantum value is attained by the following quantum strategy. Let $r=\lfloor m/2 \rfloor$, and consider local Hilbert spaces
\begin{equation}
	\mathcal H_{\mathrm A}=\mathcal H_{\mathrm B}=\mathbb C^d,
\end{equation}
where $d=2^r$. Denote by
\begin{equation}\label{eq:canonical_state}
	\ket{ \Phi_d}=\frac{1}{\sqrt d}\sum_{i=1}^d\ket{i}_{\mathrm A}\otimes \ket{i}_{\mathrm B}
\end{equation}
the maximally entangled state between Alice and Bob. Let Alice’s observables $\{\widetilde A_i\}_i$ be a family of pairwise anticommuting operators satisfying
\begin{equation}\label{eq:canonical_operator_alice}
	\{\widetilde A_i,\widetilde A_k\}=2\delta_{ik}\mathbb I, \qquad i,k=1,\dots,m,
\end{equation}
and define Bob’s observables $\{\widetilde B_j\}_j$ by
\begin{equation}\label{eq:canonical_operator_bob}
	\widetilde B_j=\sum_{i=1}^n h_{ij}\widetilde A_i^{T},\qquad j=1,\dots,n,
\end{equation}
where $\widetilde A_i^{T}$ denotes the transpose of $\widetilde A_i$. We can now introduce the following definition.
\begin{defn}\label{def:canonical_strategy}
	Let $h$ be an $m\times n$ ROCN matrix. The \emph{canonical} or \emph{reference strategy} for achieving the quantum bound $n$ of the ROCN Bell functional $I_h$ consists of the maximally entangled state~\eqref{eq:canonical_state} together with the observables defined in Eqs.~\eqref{eq:canonical_operator_alice} and~\eqref{eq:canonical_operator_bob}.
\end{defn}
It is worth noting that the anticommuting observables~\eqref{eq:canonical_operator_alice} can be explicitly represented using the Jordan–Wigner representation~\cite{Jordan1928} on the space $\bigotimes_{k=1}^r \mathbb C^2$; see also Ref.~\cite{Michalski2025}.

\subsection{Self-testing with ROCN Bell inequalities}\label{sec:rocb_self_testing} 

In certain scenarios, a set of nonlocal quantum correlations $\{p(a,b|i,j)\}$ uniquely determines the quantum state and measurement operators used by Alice and Bob. This property is known as self-testing. The notion of uniqueness is understood only up to local transformations that leave the observed correlations invariant. We formalize this in the following definition.
\begin{defn}\label{def:self_testing}
	Let a set quantum correlations $\{p(a,b|i,j)\}$ be generated by a bipartite quantum state
	\begin{equation}
		\ket*{\widetilde \psi}\in\mathbb C^d\otimes \mathbb C^d
	\end{equation}
	and by local measurement operators
	\begin{equation}
		\{\widetilde A_i\}_{i=1,\dots, m}\quad \{\widetilde B_j\}_{j=1,\dots, n}.
	\end{equation}
	We say that $\{p(a,b|i,j)\}$ self-tests the quantum state $\ket*{\widetilde \psi}$ and the measurement operators $\{\widetilde A_i\}_i$ and $\{\widetilde B_j\}_j$ if, for any state $\ket{\psi}\in\mathcal H_{\mathrm{A}}\otimes \mathcal H_{\mathrm{B}}$ and observables $\{A_i'\}_i$ and $\{B_j'\}_j$ that reproduce the same correlations $\{p(a,b|i,j)\}$, there exist local isometries
	\begin{equation}
		U_{\mathrm{A}}: \mathcal H_{\mathrm{A}} \rightarrow \mathbb C^d \otimes \mathcal H_{\mathrm{A}}',\quad U_{\mathrm{B}}: \mathcal H_{\mathrm{B}} \rightarrow \mathbb C^d \otimes \mathcal H_{\mathrm{B}}'
	\end{equation}
	such that 
	\begin{equation}
		U_{\mathrm A}\otimes U_{\mathrm B}\ket{\psi'}=\ket{\psi}\otimes \ket*{\varphi_{\mathrm{junk}}}
	\end{equation}
	and 
	\begin{equation}
		U_{\mathrm A}A_iU_{\mathrm A}^\dagger = \widetilde  A_i \otimes\mathbb I_{\mathrm A'},\quad U_{\mathrm B}B_jU_{\mathrm B}^\dagger =\widetilde B_j \otimes\mathbb I_{\mathrm B'}.
	\end{equation}
	Here, $\mathcal H_{\mathrm A}'$ and $\mathcal H_{\mathrm B}'$ are auxiliary Hilbert spaces, and $\ket*{\varphi_{\mathrm{junk}}}$ denotes a normalized state supported on $\mathcal H_{\mathrm A}'\otimes\mathcal H_{\mathrm B}'$.
\end{defn}
Suppose that the maximal quantum value $\beta_Q^h$ of a Bell functional $I_h$ can be achieved only by a specific set of quantum correlations $\{p(a,b|i,j)\}$, and that this correlation self-tests a particular quantum strategy. In that case, we say that the Bell inequality itself self-tests the corresponding quantum strategy. 

As an example, consider the CHSH Bell functional
\begin{equation}\label{eq:chsh}
	I_{\mathrm{CHSH}}:=\langle A_1B_1\rangle+\langle A_1B_2\rangle+\langle A_2B_1\rangle-\langle A_2B_2\rangle,
\end{equation}
whose maximal quantum value is $2\sqrt 2$. This value uniquely identifies (up to local isometries) the maximally entangled two-qubit state
\begin{equation}
	\ket{\Phi_2}=\frac{1}{\sqrt 2} \left(\ket{1}\otimes \ket{1}+ \ket{2}\otimes \ket{2}\right)
\end{equation}
together with the local measurements
\begin{align}
	A_1 =\sigma_z&,\quad A_2=\sigma_x,\\
	B_1 = \frac{1}{\sqrt 2} (\sigma_z+\sigma_x)&,\quad B_2 = \frac{1}{\sqrt 2} (\sigma_z-\sigma_x).
\end{align}
Hence, the CHSH inequality self-tests both the maximally entangled bipartite qubit state and the corresponding measurement operators realizing the maximal quantum violation.

For ROCN Bell inequalities, a necessary and sufficient condition for self-testing is provided by the following Theorem~\cite{Michalski2025}.

\begin{thm}[Self-testing from ROCN Bell inequalities]\label{th:self_testing_rocn}
	Let $h$ be a $m \times n$ ROCN matrix, with $m$ even. Define the matrix
	\begin{equation}\label{eq:M_matrix}
		M_{j,(i,k)}=h_{ji}h_{jk},\quad  j=1,\dots,m,\quad  1\leqslant i < k\leqslant n,
	\end{equation}
	where each pair $(i,k)$ is treated as a single, unique column index. Then, the maximal violation of the Bell inequality $I_h$ self-tests the canonical strategy introduced in Definition~\ref{def:canonical_strategy} if and only if $M$ has full column rank:
	\begin{equation}\label{eq:necessary_self_testing}
		\mathrm{rank}(M) = \frac{1}{2}{n(n-1)}.
	\end{equation}
\end{thm}

\begin{rem}
	Observe that for the matrix to have full column rank, the number of rows must be greater than or equal to the number of columns of $M$. Furthermore, Eq.~\eqref{eq:hmatrix_1} implies that the rows of $M$ are linearly dependent. Therefore, in order to obtain the rank~\eqref{eq:necessary_self_testing}, the number of rows must strictly exceed the number of columns:
	\begin{equation}\label{eq:necessary_self_testing_2}
		n > \frac{1}{2}m(m-1).
	\end{equation}
	This inequality imposes a constraint on the number of observables required for each observer to certify the measurements.
\end{rem}

The case where $m$ is odd is more subtle. In particular, Alice's operators can be transformed as
\begin{equation}
	A_j\mapsto A_j \quad \text{for}\quad j=1,\dots, m-1,\quad A_m\mapsto -A_m,
\end{equation}
without altering the anticommutation relations among the operators. Consequently, two distinct and nonequivalent strategies can give rise to the same quantum correlations. This subtlety and its implications for self-testing are discussed in detail in Ref.~\cite{Michalski2025}.

\section{Symmetric subspaces}\label{sec:symmetric_subspace}

Symmetric subspaces arise naturally in the description of bosonic systems, reflecting the indistinguishability of identical particles~\cite{Sakurai1986}. They are fundamental in various branches of physics. For example, collective emission phenomena such as superradiance and subradiance can be formulated in terms of Dicke states, which form a distinguished basis of the symmetric subspace~\cite{Dicke1954,viggiano2023}. The entanglement properties of symmetric states have been extensively explored: several entanglement criteria become equivalent when restricted to symmetric states~\cite{Toth2009}, and for certain families of symmetric states, separability is fully characterized by the PPT criterion~\cite{quesada}. However, this equivalence does not hold for all symmetric states~\cite{Augusiak_2012,Tura2018,romero_palleja}. Furthermore, special classes of symmetric states, such as GHZ and Dicke states~\cite{Supic2018,Fadel2017}, are known to be self-testable, as are antisymmetric Slater states~\cite{Konderak2025}. See also Ref.~\cite{Marconi2025} for a review on symmetric quantum systems. We now revisit the main definitions of symmetric subspaces for bipartite systems and construct explicit examples of subspaces spanning the symmetric sector.

Let $\mathcal H_A=\mathbb C^m$, and define the symmetric projector, which acts on a simple tensor $\ket{u}\otimes \ket{v}$ as
\begin{equation}
	\mathcal S(\ket{u}\otimes\ket{v})=\frac{1}{2}\left(\ket{u}\otimes\ket{v}+\ket{v}\otimes \ket{u}\right).
\end{equation}
The operator $\mathcal S$ extends linearly to all elements of the tensor product space $\mathbb C^m\otimes \mathbb C^m$. Its range is the symmetric subspace $\mathcal S(\mathbb C^m\otimes \mathbb C^m)$. For two normalized vectors $\ket{v}$ and $\ket{w}$, the symmetric tensor product $\ket{v}\vee\ket{w}$ is defined as
\begin{equation}
	\ket{v}\vee \ket{w}=\frac{1}{\norm{\mathcal S(\ket{u}\otimes \ket{v})}}\mathcal S(\ket{u}\otimes \ket{v}).
\end{equation}
Let $\{\ket{e_i}\}_{i=1}^m$ be an orthonormal basis of $\mathbb C^m$. Then, the following family of vectors span all the symmetric subspace:
\begin{equation}
	\{{\ket{e_i} \vee \ket{e_k}}:1\leqslant i\leqslant k\leqslant m\}.
\end{equation}
These vectors are linearly independent, and hence the dimension of the symmetric subspace $\mathcal S(\mathbb C^m\otimes \mathbb C^m)$ is $m(m+1)/2$.
For any $i<k$, we observe that
\begin{equation}
	2\mathcal{S}(\ket{e_i} \otimes \ket{e_k}) = (\ket{e_i}+\ket{e_k})\otimes(\ket{e_i}+\ket{e_k})-\ket{e_i}{\otimes \ket{e_i}}-\ket{e_k}{\otimes \ket{e_k}}.
\end{equation}
Introducing the shorthand notation $\ket{u}^{\otimes 2}=\ket{u}\otimes \ket{u}$, we find that the set of vectors
\begin{equation}\label{eq:spanning_2_fold_product}
	\{\left(\ket{e_i}+\ket{e_k}\right)^{\otimes 2}:1\leqslant i\leqslant k\leqslant m\}
\end{equation}
also forms a basis of the symmetric subspace $\mathcal S(\mathbb C^m\otimes \mathbb C^m)$. We define a set of vectors $\{\ket{v_i}\}_i \subset \mathbb C^m$ to be a symmetric spanning set if the two-fold tensor products $\{\ket{v_i}^{\otimes 2}\}_i$ span the whole symmetric subspace $\mathcal S(\mathbb C^m\otimes \mathbb C^m)$~\cite{Procesi2007}. As we will see in the following section, symmetric spanning sets can be used to construct ROCN matrices that self-test the canonical strategy introduced in Definition~\ref{def:canonical_strategy}.

For this purpose, for each $i=1,\dots,m$, consider the family of vectors
\begin{equation}\label{eq:basis}
	\left\{\frac{\ket{e_i}+\ket*{e_{i\oplus k}}}{\sqrt{2}}:k=1,\dots,m-1\right\}\cup \{\ket{e_i}:i=1,\dots,m\},
\end{equation}
where $i\oplus k$ denotes the sum modulo $m$. Each vector in Eq.~\eqref{eq:basis} is normalized, though the vectors are not necessarily mutually orthogonal. We then define
\begin{equation}\label{eq:a_vectors}
	\left\{\ket{a_{ik}}:k=1,\dots,m\right\}
\end{equation}
as the orthonormal system obtained from~\eqref{eq:basis} via the Gram--Schmidt orthonormalization procedure. This construction allows us to formulate the following result.

\begin{thm}\label{th:set_spanning}
	Let $\{\ket{e_i}\}_{i=1}^m$ be an orthonormal basis in $\mathbb C^m$, and, for each $i=1,\dots,m$, let $\{\ket{a_{ik}}\}_{k=1}^m$ be defined as in Eq.~\eqref{eq:a_vectors}. Then, the set of vectors
	\begin{equation}\label{eq:set_spanning}
		\left\{\ket{e_i}: 1 \leqslant i \leqslant m \right\}
		\ \cup
		\left\{\ket{a_{ik}}: 1 \leqslant i \leqslant k \leqslant m\right\}
	\end{equation}
	forms a symmetric spanning set
\end{thm}
\begin{proof}
	We first express each vector $\ket{a_{ik}}$ recursively in the basis $\{\ket{e_j}\}_j$:
	\begin{align}\label{eq:vectors_matrix}
		\ket{a_{i1}}&=\frac{\ket{e_i}+\ket*{e_{i\oplus 1}}}{\sqrt{2}},\nonumber\\
		\ket{a_{i2}}&={\frac{1}{\sqrt 6}}\ket{e_i}-{\frac{1}{\sqrt 6}}\ket*{e_{i\oplus 1}}+\sqrt{\frac{2}{3}}\ket{e_{i\oplus 2}},\nonumber\\
		&\vdots \nonumber\\
		\ket{a_{ik}}&= \alpha_{k}\left(\ket{e_i}-\sum_{j=1}^{k-1}\ket*{e_{i\oplus j}}\right)+\beta_{k}\ket{e_{i\oplus k}}.
	\end{align}
	This structure arises because each $\ket{a_{ik}}$ is orthogonal to $\ket{e_i}+\ket*{e_{i\oplus j}}$ for all $1\leqslant j<k$. Explicit calculation via the Gram--Schmidt procedure shows that
	\begin{equation}\label{alpha_beta}
		\alpha_k=\frac{1}{\sqrt{k+k^2}},\quad \beta_k=\sqrt{\frac{k}{1+k}}.
	\end{equation} 
	Indeed, for the vector $(\ket{e_i}+\ket*{e_{i\oplus k}})/\sqrt{2}$, the orthogonalization step reads
	\begin{equation}
		\frac{\ket{e_i}+\ket*{e_{i\oplus k}}}{\sqrt 2}-\sum_{j=1}^{k-1}\ketbra*{a_{ij}}\frac{\ket{e_i}+\ket*{e_{i\oplus k}}}{\sqrt 2}=	\frac{\ket{e_i}+\ket*{e_{i\oplus k}}}{\sqrt 2}-\frac{1}{\sqrt 2}\sum_{j=1}^{k-1}\alpha_j\ket*{a_{ij}}.
	\end{equation}
	Comparing this with the general form of $\ket{a_{ik}}$ yields the relation
	\begin{equation}
		\frac{\alpha_k}{\beta_k} = 1+\sum_{j=1}^{k-1}\alpha_j^2.
	\end{equation}
	Together with the normalization condition $k\alpha_k^2+\beta_k^2=1$, this gives Eq.~\eqref{alpha_beta} by induction.
	
	Next, we show that the two-fold tensor products of the vectors in Eq.~\eqref{eq:set_spanning} span the symmetric subspace. Write $\ket{a_{ik}}=\sum_{j=0}^k \gamma_{kj} \ket*{e_{i\oplus j}}$, so that
	\begin{align}
		\ket{a_{ik}}^{\otimes 2}&=\sum_{j,\ell=0}^k\gamma_{kj}\gamma_{k\ell} \ket*{e_{i\oplus j}}\otimes \ket*{e_{i\oplus \ell}}\nonumber\\ &= \sum_{\j=1}^k\ \gamma_{kj}^2\ket{e_{i\oplus j}}^{\otimes 2}+2\sum_{j<\ell}^{k} \gamma_{kj}\gamma_{k\ell} \mathcal S(\ket*{e_{i\oplus j}}\otimes \ket{e_{i\oplus \ell}}).\label{eq:2_fold_a_ik}
	\end{align}
	We now show by induction on $k$ that each symmetric tensor $\mathcal S(\ket{e_i}\otimes\ket{e_{i\oplus k}})$ lies in the span of the set in Eq.~\eqref{eq:set_spanning}. For the base case of $k=1$ and arbitrary $i$, one directly obtains
	\begin{equation}
		2\mathcal S(\ket{e_{i}}\otimes \ket{e_{i\oplus {1}}})=\ket{a_{i1}}^{\otimes 2} - \ket{e_i}^{\otimes 2}-\ket{e_{i\oplus 1}}^{\otimes 2},
	\end{equation}
	which is clearly in the span of~\eqref{eq:set_spanning}.
	Assume now that $\mathcal S(\ket{e_{i}}\otimes \ket*{e_{i\oplus j}})$ lies in the span of the set in Eq.~\eqref{eq:set_spanning} for all $1\leqslant i\leqslant m$ and $1\leqslant j\leqslant k-1$. From Eq.~\eqref{eq:2_fold_a_ik}, we can isolate the term $\mathcal S(\ket{e_i}\otimes\ket{e_{i\oplus k}})$:
	\begin{align}
		\alpha_k\beta_k\mathcal S(\ket{e_{i}}\otimes \ket{e_{i\oplus k}}) &= \frac{1}{2}\left[{\ket{a_{ik}}^{\otimes 2}}- \sum_{\j=1}^k\ \gamma_{kj}^2\ket{e_{i\oplus j}}^{\otimes 2}\right]\nonumber\\ \quad &-\sum_{{\substack{j<\ell\\(j,\ell)\neq(0,k)}}}\gamma_{kj}\gamma_{k\ell} \mathcal S(\ket*{e_{i\oplus j}}\otimes \ket*{e_{i\oplus \ell}}).
	\end{align}
	By the induction hypothesis, all terms on the right-hand side are in the span of the set~\eqref{eq:set_spanning}, and therefore so is $\mathcal S(\ket{e_i}\otimes\ket{e_{i\oplus k}})$. Hence, by induction, the set in Eq.~\eqref{eq:set_spanning} spans the entire symmetric subspace $\mathcal S(\mathbb C^m \otimes \mathbb C^m)$.
\end{proof}

We note that the total number of elements in~\eqref{eq:set_spanning} is $m(m+1)$, which exceeds the dimension $\tfrac{1}{2}m(m+1)$ of the symmetric subspace; thus, the system is not minimal, and the elements are linearly dependent.

\section{Sufficient condition for self-testing of an ROCN Bell inequality}\label{sec:sufficient_condition}

We are now ready to establish a connection between the symmetric subspace and the self-testing condition for ROCN matrices. As shown in Theorem~\ref{th:self_testing_rocn}, self-testing of a given ROCN inequality is completely determined by the rank of the associated matrix $M$. Although the criterion~\eqref{eq:necessary_self_testing} provides a necessary and sufficient condition, its formulation may appear somewhat abstract from a physical perspective. To obtain a more transparent and operational viewpoint, we now derive an equivalent characterization expressed directly in terms of the column vectors of the defining matrix $h$.

\begin{thm}\label{th:self_testing_column_vectors}
	Let $h$ be an $m\times n$ ROCN matrix with $m$ even, and let $\{\ket{h_j}\}_{j}\subset\mathbb C^m$ denote its column vectors. Then the maximal quantum value of the ROCN Bell functional $I_h$ self-tests the canonical strategy introduced in Definition~\ref{def:canonical_strategy} if and only if, for every symmetric vector $\ket{\psi} \in \mathcal S(\mathbb C^m \otimes \mathbb C^m)$ with null diagonal elements, $(\bra{e_i}\otimes\bra{e_i})\ket{\psi}=0$ for $i = 1,\dots,m$, the following condition holds:
	\begin{equation}\label{eq:necessary_self_testing_3}
		\forall\; j =1,\dots, n: \quad \left(\bra{h_j}\otimes \bra{h_j}\right) \ket{\psi}=0 \quad \Leftrightarrow \quad \ket{\psi}=0.
	\end{equation}
	In particular, if $\mathcal S(\mathbb C^m\otimes\mathbb C^m)=\mathrm{span}\{\ket{h_j}\otimes \ket{h_j}:j=1,\dots, {n}\}$, then the maximal quantum value of the ROCN Bell functional $I_h$ self-tests the canonical strategy of Definition~\ref{def:canonical_strategy}.
\end{thm}
\begin{proof}
	First, observe that each column vector $\ket{h_j}$ of $h$ can be written as
	\begin{equation}
		\ket{h_j}=\sum_{i=1}^{m}h_{ij}\ket{i}\in\mathbb C^m,
	\end{equation}
	where $\{\ket{i}\}_{i=1}^m$ is the canonical basis of $\mathbb C^m$. From Eq.~\eqref{eq:hmatrix_2}, it follows that the vectors $\ket{h_j}$ are normalized.
	By Theorem~\ref{th:self_testing_rocn}, the ROCN Bell functional $I_h$ self-tests the canonical strategy if and only if the matrix $M$ defined in Eq.~\eqref{eq:M_matrix} has full column rank. Equivalently, this means that for any family of complex numbers $\{O_{ik}\}_{ij}$, with $1 \leqslant i < k \leqslant m$, the following holds:
	\begin{equation}
		\forall\; j=1,\dots,{n} :\quad \sum_{i<k}^m h_{ij} h_{kj} O_{ik} = 0 \quad \Leftrightarrow \quad \forall\;{i < k} : \quad O_{ik} = 0.
	\end{equation}
	We now rewrite this criterion in a matrix-vector form. We require that for every complex symmetric matrix $O \in \mathbb{C}^{m \times m}$ with null diagonal elements, the following holds:
	\begin{equation}\label{eq:self_testin_sufficient_1}
		\forall\; j=1,\dots,{n} :\quad\sum_{i,k=1}^m h_{ij}h_{kj}O_{ik}=0 \quad\Leftrightarrow\quad O=0.
	\end{equation}
	Let us define
	\begin{equation}
		\ket{{\psi}}=\sum_{i,k=1}^{m}O_{ik}\ket{i} \otimes \ket{k}\in\mathcal S(\mathbb C^m\otimes \mathbb C^m).
	\end{equation}
	Then, condition~\eqref{eq:self_testin_sufficient_1} can be rewritten as~\eqref{eq:necessary_self_testing_3}. The last part of the theorem is a consequence of the non-degeneracy of the scalar product.
\end{proof}

Note that the equality $\mathcal S(\mathbb C^m\otimes\mathbb C^m)=\mathrm{span}\{\ket{h_j}\otimes \ket{h_j}:j=1,\dots, {n}\}$ implies that the $n$ vectors $\ket{h_j} \otimes \ket{h_j}$ span a Hilbert space of dimension ${m(m+1)}/2$. Consequently, in this case, $n\geqslant {m(m+1)}/2>m(m-1)/2$, and we get a stronger condition than the rank condition in Eq.~\eqref{eq:necessary_self_testing}.

This symmetric-spanning condition is sufficient but not necessary. For example, the CHSH functional~\eqref{eq:chsh} is defined by the ROCN matrix
\begin{equation}
	h_{\mathrm{CHSH}}=\frac{1}{\sqrt 2}
	\begin{pmatrix}
		+1 & +1 \\ +1 & -1
	\end{pmatrix}.
\end{equation}
Its column vectors $\ket{h_1}$ and $\ket{h_2}$ are
\begin{equation}
	\ket{h_1}=\frac{1}{\sqrt{2}} \begin{pmatrix}
		+1 \\ +1
	\end{pmatrix},\quad \ket{h_2}=\frac{1}{\sqrt{2}} \begin{pmatrix} +1\\ -1
	\end{pmatrix}.
\end{equation}
Since the symmetric two-qubit subspace has dimension three, these two column vectors cannot span the symmetric subspace; hence the CHSH columns do not form a symmetric spanning set, even though CHSH does self-test the canonical qubit strategy. This example illustrates that the symmetric-spanning criterion is stronger than necessary for self-testing in general.

\section{Construction of self-testing ROCN Bell inequalities}\label{sec:self-testing_construction}

As we established in Section~\ref{sec:sufficient_condition}, a sufficient condition for an ROCN matrix $h$ to enable self-testing is that its column vectors form a symmetric spanning set. Our goal in this section is therefore to identify a suitable symmetric spanning family from which an ROCN matrix can be constructed. 

At first sight, the set in Eq.~\eqref{eq:spanning_2_fold_product} appears to be a natural choice: its two-fold tensor products span the entire symmetric subspace. However, these vectors cannot be directly used as the columns of $h$, since they fail to satisfy the {normalization} requirement~\eqref{eq:hmatrix_2} in the definition of an ROCN matrix. The family introduced in Theorem~\ref{th:set_spanning}, obtained by applying the Gram--Schmidt procedure to the vectors in Eq.~\eqref{eq:basis}, circumvents this problem. Each resulting vector is properly normalized, and the entire collection forms a symmetric spanning set. As we show below, this allows us to construct ROCN Bell inequalities that self-test any even number of Clifford generators~\eqref{eq:canonical_operator_alice}, with local dimension $m=2r$. The main result is summarized in the following theorem.

\begin{thm}\label{th:self_testing_construction}
	For every even integer $m=2r$, there exists an ROCN matrix $h$ of dimension $m\times m(m+1)$ that enables self-testing of the canonical strategy introduced in Definition~\ref{def:canonical_strategy}.
\end{thm}

\begin{proof}
	The matrix $h$ is constructed starting from a family of orthogonal matrices
	\begin{equation}
		\{O^{(0)},O^{(1)},\dots,O^{(m)}\},
	\end{equation}
	each of dimension $m\times m$, by concatenation:
	\begin{equation}\label{eq:h_matrix_example}
		h=[O^{(0)}\  O^{(1)} \dots\  O^{(m)}].
	\end{equation}
    Explicitly, let $1\leqslant i\leqslant m$ and $1\leqslant j \leqslant m(m+1)$. Then, suppose $\ell m+1 \leqslant j \leqslant (\ell+1)m$ for some $\ell=0,\dots, m$, and define
	\begin{equation}
		h_{ij}=O_{i,j-\ell m}^{(\ell)}.
	\end{equation} 
    By construction, all columns of $h$ have unit norm and it rows are normalized and mutually orthogonal. Hence $h$ satisfies both ROCN conditions and it is semi-orthogonal. 
    
    We now specify the blocks $O^{(\ell)}$. Let $\{\ket{e_i}\}_{i=1}^m$ be the canonical real orthonormal basis of $\mathbb C^m$. We set the columns of $O^{(0)}$ as these basis vectors:
	\begin{equation}
		O^{(0)}=[\ket{e_1}\ \ket{e_2}\  \dots \ket{e_m}].
	\end{equation}
	Next, recall the vectors $\ket{a_{ik}}$ introduced in Eq.~\eqref{eq:a_vectors}. For each fixed $i$, the set $\{\ket{a_{ik}}:k=1,\dots,m\}$ forms an orthonormal basis of $\mathbb C^m$.  Therefore, for each $\ell = 1, \ldots, m$, we define
	\begin{equation}
		O^{(\ell)}=[\ket*{a_{\ell 1}}\ket*{a_{\ell 2}}\dots \ket*{a_{\ell m}}].
	\end{equation}
	Each $O^{(\ell)}$ is {orthogonal} by construction, and thus the concatenated matrix $h$ given in Eq.~\eqref{eq:h_matrix_example} satisfies both ROCN conditions. Finally, by Theorem~\ref{th:set_spanning}, the columns of $h$ form a symmetric spanning set. By Theorem~\ref{th:self_testing_column_vectors}, any ROCN matrix whose columns form a symmetric spanning set ensures self-testing of the canonical strategy. This completes the proof.
\end{proof}


\section{Examples}\label{sec:examples}

In this section, we provide two examples of self-testing ROCN matrices using the construction outlined in Theorem~\ref{th:self_testing_construction}. We consider the cases $m=2$ and $m=4$.

\begin{exmp}[$m=2$]
	Consider the case of two anticommuting observables. The three orthogonal matrices introduced in the proof of Theorem~\ref{th:self_testing_construction} are
	\begin{equation}
		O^{(0)}=\begin{pmatrix}
			+1 & 0 \\ 0 & +1
		\end{pmatrix},\quad O^{(1)}=\frac{1}{\sqrt 2}\begin{pmatrix}
			+1 & +1 \\ +1 & - 1
		\end{pmatrix},\quad O^{(2)} = \frac{1}{\sqrt 2}\begin{pmatrix}
			+1 & -1 \\ +1 & +1
		\end{pmatrix},
	\end{equation}
	with respect to the canonical basis of $\mathbb C^2$. In this case, Alice will have access to two observables $\{A_1,A_2\}$, while Bob will have access to six observables $\{B_1,B_2,\dots,B_6\}$. All together, the Bell functional is in the form
    \begin{align}
        I_h&=\frac{A}{\sqrt 2} A_1 \left[ \sqrt 2 B_1 + B_3 +B_4 - B_5 - B_6 \right]\nonumber\\
        &+ \frac{1}{\sqrt 2}A_2\left[\sqrt{2}B_2+B_3- B_4 + B_5 + B_6\right].
    \end{align}
    Observe that the matrix $O^{(1)}$ already defines the CHSH Bell inequality. 
    Moreover, the second and third matrices share the same column vectors and therefore do not contribute to self-testing or to spanning the symmetric subspace. This explicitly shows that the construction provided is not optimal. Consequently, only $O^{(0)}$ and $O^{(1)}$ are necessary to span the symmetric subspace. In this sense, the construction given by Theorem~\ref{th:set_spanning} is overdetermined.
\end{exmp}

\begin{exmp}[$m=4$]
	We now apply the same construction as provided in the proof of Theorem~\ref{th:self_testing_construction} for the case of $m=4$. Working in the canonical basis, the orthogonal matrices are
	\begin{align}
		O^{(0)}&=\begin{pmatrix}
			1 & 0 &0 &0 \\ 0 & 1 & 0 & 0 \\ 0& 0 & 1 & 0 \\ 0& 0 & 0& 1
		\end{pmatrix},\quad O^{(1)}=\begin{pmatrix}
			1/\sqrt{2} & 1/\sqrt{6} & 1/2\sqrt 3 & 1/2 \\ 1/\sqrt{2} & -1/\sqrt{6} & -1/2\sqrt 3  & -1/2 \\ 0& \sqrt{2/3} & -1/2\sqrt 3  & -1/2 \\ 0& 0 & \sqrt{3/2}& -1/2
		\end{pmatrix},\nonumber\\
		O^{(2)}&=\begin{pmatrix}
			0& 0 & \sqrt{3/2}& -1/2\\
			1/\sqrt{2} & 1/\sqrt{6} & 1/2\sqrt 3 & 1/2 \\ 1/\sqrt{2} & -1/\sqrt{6} & -1/2\sqrt 3  & -1/2 \\ 0& \sqrt{2/3} & -1/2\sqrt 3  & -1/2
		\end{pmatrix},\nonumber\\
		O^{(3)}&=\begin{pmatrix}
			0& \sqrt{2/3} & -1/2\sqrt 3  & -1/2\\
			0& 0 & \sqrt{3/2}& -1/2\\
			1/\sqrt{2} & 1/\sqrt{6} & 1/2\sqrt 3 & 1/2 \\ 1/\sqrt{2} & -1/\sqrt{6} & -1/2\sqrt 3  & -1/2
		\end{pmatrix},\nonumber\\
		O^{(4)}&=\begin{pmatrix}
			1/\sqrt{2} & -1/\sqrt{6} & -1/2\sqrt 3  & -1/2 \\
			0& \sqrt{2/3} & -1/2\sqrt 3  & -1/2\\
			0& 0 & \sqrt{3/2}& -1/2\\
			1/\sqrt{2} & 1/\sqrt{6} & 1/2\sqrt 3 & 1/2
		\end{pmatrix}.
	\end{align}
    In particular, the $j$-column of the matrix $O^{(i)}$ is given by the vectors $\ket{a_{ij}}$ of Eq.~\eqref{eq:vectors_matrix}, with $\alpha_k$ and $\beta_k$ as in Eq.~\eqref{alpha_beta}. The orthogonality condition of the matrices $O^{(i)}$ ensures ROCN conditions~\eqref{eq:hmatrix_1} and~\eqref{eq:hmatrix_2} are satisfied. Alice has therefore $4$ anticommuting observables $\{A_1,\dots,A_4\}$, while Bob has $20$ observables $\{B_1,\dots,B_{20}\}$. The Bell functional is therefore in the form
    \begin{equation}
        I_h=\sum_{\ell=0}^4 \sum_{ij=1}^4 A_{i} O_{ij}^{(\ell)} B_{i+j}.
    \end{equation}
	As in the $m=2$ case, the problem is overdetermined, and the first matrix $O^{(0)}$ can be omitted, as its columns do not contribute to condition~\eqref{eq:necessary_self_testing_3}, according to Theorem~\ref{th:self_testing_column_vectors}.In particular, this appears to hold across all dimensions when the canonical basis is chosen. 
\end{exmp}

\section{Conclusions}
Self-testing constitutes one of the strongest forms of certification of quantum resources. It is a fully device-independent technique: no assumptions are made about the internal functioning of the devices, and yet, in certain scenarios, the maximal violation of a Bell inequality uniquely determines the underlying quantum state and measurements. The class of ROCN Bell inequalities, introduced in Ref.~\cite{Michalski2025}, provides a broad family of such examples. These inequalities generalize the CHSH and Elegant Bell inequalities, admit an analytic characterization of the quantum bound, and come equipped with a necessary and sufficient condition for self-testing.

In this paper, we explored explicit constructions of ROCN Bell inequalities and identified a sufficient condition for self-testing of the canonical strategy. In particular, we showed that it is sufficient for the column vectors to form a symmetric spanning set, meaning that their two-fold tensor products span the entire symmetric subspace. Building on this criterion, we presented an explicit family of ROCN Bell inequalities that self-test the canonical strategy associated with any even number of Clifford generators. While this construction is general and conceptually simple, it is not optimal: it uses significantly more columns than the minimal number required for self-testing. This observation points to an interesting direction for future research---namely, the design of more efficient constructions that achieve self-testing with a minimal number of columns while preserving the analytic structure of ROCN inequalities.

\section{Acknowledgments}

We thank Remigiusz Augusiak and Wojciech Bruzda for valuable discussions. This work is supported by the National Science Centre (Poland) through the SONATA BIS project No.\ 019/34/E/ST2/00369, and funding from the European Union's Horizon Europe research and innovation programme under grant agreement No.\ 101080086 NeQST.

\endpaper
\end{document}